\RequirePackage[standard,thmmarks,thref]{ntheorem}
\documentclass[submission,copyright,creativecommons]{eptcs}
\providecommand{\event}{LSFA 2024} 

\usepackage{iftex}

\ifpdf
  \usepackage{underscore}         
  \usepackage[T1]{fontenc}        
\else
  \usepackage{breakurl}           
\fi

\newcommand{\Linf}{\ensuremath{\text{\L}_\infty}}
\newcommand{\Luka}{{\L}ukasiewicz}

\newcommand{\inte}[1]{#1^\circ}
\newcommand{\tuple}[1]{\langle{#1}\rangle}
\renewcommand{\phi}{\varphi}
\renewcommand{\emptyset}{\varnothing}
\usepackage{amsfonts,amsmath,amssymb}
\usepackage{tikz,pgfplots}
\usepgfplotslibrary{patchplots}
\usepackage{subcaption}
\usepackage{xfrac}
\usepackage{algorithm}
\usepackage{algorithmic}

\title{Regional, Lattice and Logical Representations\\
	of Neural Networks}
\author{
	Sandro Preto
	\institute{Center for Mathematics, Computing and Cognition\\
		Federal University of ABC, Brazil}
	\institute{
	Institute of Mathematics and Statistics\\
		University of São Paulo, Brazil}
	\email{sandro.preto@ufabc.edu.br}
	\and
	Marcelo Finger
	\institute{Institute of Mathematics and Statistics\\
	University of São Paulo, Brazil}
	\email{mfinger@ime.usp.br}
}

\newcommand{\titlerunning}{Regional, Lattice and Logical Representations of Neural Networks}
\newcommand{\authorrunning}{S. Preto \& M. Finger}

\hypersetup{
  bookmarksnumbered,
  pdftitle    = {\titlerunning},
  pdfauthor   = {\authorrunning},
  pdfsubject  = {EPTCS},               
}

\begin{document}
\maketitle

\begin{abstract}
A possible path to the interpretability of neural networks is to (approximately) represent them in the regional format of piecewise linear functions, where regions of inputs are associated to linear functions computing the network outputs. We present an algorithm for the translation of feedforward neural networks with ReLU activation functions in hidden layers and truncated identity activation functions in the output layer. We also empirically investigate the complexity of regional representations outputted by our method for neural networks with varying sizes. Lattice and logical representations of neural networks are straightforward from regional representations as long as they satisfy a specific property. So we empirically investigate to what extent the translations by our algorithm satisfy such property.
\end{abstract}

\section{Introduction}
Neural networks are computational models that aim to generalize patterns found in datasets from which they are determined by means of a learning algorithm \cite{GBC2016}. Despite the undeniable advancement in the state of the art of intelligent systems promoted by neural networks, their lack of interpretability is subject to criticism. Neural networks suffer from the \emph{black box problem} due to the lack of justification for their results and the impossibility to directly inspect their learned information \cite{Cas2016,Fin2020}.

As several architectures of neural networks realize piecewise linear functions or may be approximated by them, a path towards interpretability is through \emph{regional format} representations of such neural networks and functions by explicit sets of pairs $\tuple{p,\Omega}$ of a linear piece $p$ and a region $\Omega$ such that, for a vector of input values $\mathbf{x}\in\Omega$, the output is given by $p(\mathbf{x})$. An algorithm for establishing regional representations from feedforward neural networks with rectified linear units as activation functions is proposed in \cite{RRS2019}.

The main goal of this work is to introduce an algorithm for computing regional format representations of \emph{ReLU--TId neural networks}, which are feedforward neural networks with rectified linear units as activation functions in hidden layers and truncated identity as activation functions in the output layer. Such algorithm outputs representations in the \emph{pre-closed regional format}, where regions are polyhedra. Rather than just adapting the iterative method in \cite{RRS2019}, we present a novel recursive approach that allows a correctness proof by a straightforward induction argument.

An important feature of neural networks is that they are compact representations of functions. Then, although regional representations might provide interpretability of neural networks, they also might be exponential in the size of their traditional representation as graphs. In Section \ref{sec:experiments}, we empirically measure the complexity of regional representations determined by our method for randomly generated ReLU--TId neural networks with varying numbers of neurons and layers and varying layer sizes.

Lattice representation is another possibility for representing neural networks and is achieved by combining maximum and minimum operations over linear pieces. Such representations further enable the codification of ReLU--TId neural networks in logical systems as \Luka\ infinitely-valued logic (\Linf) and its extensions \cite{CDM2000,Ger2001,PF2020,PF2022}, leading to yet another path to interpretability. Lattice and logical representations find applications in the formal verification of neural networks in attempts to circumvent the black box problem and allow their use in critical tasks; for instance, in aircraft collision avoidance alerts and autonomous vehicles. There are methods for formal verification using the lattice representation of neural networks \cite{APS2023} and methods that codify properties of neural networks in the language of \Linf\ departing from their logical representation \cite{PF2023b,PF2023}.

Lattice and logical representations may be built in polynomial time from ReLU--TId neural networks given in the pre-closed regional format as long as such encodings satisfy the so-called \emph{lattice property} (Section \ref{sec:preliminaries}) \cite{PF2022}. In this case, the regional representations are said to be in the \emph{closed regional format}. This work also aims at empirically experimenting how far from satisfying lattice property are randomly generated neural networks.

The rest of this work is organized as follows. Section \ref{sec:preliminaries} introduces neural networks and their graph, regional, lattice and logical representations. Section \ref{sec:nn2pwl} presents an algorithm for translating ReLU--TId neural networks into the pre-closed regional format. Section \ref{sec:experiments} presents the results of experiments where we measure the complexity of representations in pre-closed regional format and their degree of satisfiability of lattice property.

\section{Preliminaries: Some Neural Networks and Their Representations}
\label{sec:preliminaries}
Traditionally, a \emph{feedforward neural network} $N$ is given (and represented) by a graph whose nodes are partitioned into an ordered family of ordered sets $\mathcal{L}_N = \{ L_0,\ldots,L_\Lambda \}$, where each $L_i$ is a \emph{layer}. All nodes in layer $L_i$, for $i\in \{0,\ldots,\Lambda-1\}$, are linked by an edge to all nodes in layer $L_{i+1}$ establishing a computational circuit such that all output values of nodes in $L_i$ are input values to each node in $L_{i+1}$. There is a linear function $f_j^i:\mathbb{R}^{|L_{i-1}|}\to\mathbb{R}$ associated to each node $n_j^i$ in layer $L_i$, for $j\in \{1,\ldots,|L_i|\}$ and $i\in \{1,\ldots,\Lambda\}$. For a tuple of input values $\mathbf{x}=\tuple{x_1,\ldots,x_{|L_{i-1}|}}\in\mathbb{R}^{|L_{i-1}|}$ to node $n_j^i$, it has as output the value $n_j^i(\mathbf{x}) = \rho_i\circ f_j^i(\mathbf{x})$, where $\rho_i:\mathbb{R}\to\mathbb{R}$ is an \emph{activation function}. Thus, for $\mathbf{x}=\tuple{x_1,\ldots,x_{|L_{i-1}|}}$ as input to layer $L_i$, it has as outputs the values in the tuple
\[ L_i(\mathbf{x})=\tuple{\rho_i\circ f_1^i(\mathbf{x}),\ldots,\rho_i\circ f_{|L_i|}^i(\mathbf{x})}. \]
Input values to $N$ in a tuple $\mathbf{x}=\tuple{x_1,\ldots,x_{|L_0|}} \in \mathbb{R}^{|L_0|}$ are neatly associated to the nodes $n_1^0,\ldots,n_{|L_0|}^0 \in L_0$, called \emph{input nodes}; thus, from such inputs, $N$ produces the output values in the $|L_\Lambda|$-tuple
\[ N(\mathbf{x}) = \tuple{N(\mathbf{x})_1,\ldots,N(\mathbf{x})_{|L_\Lambda|}} = L_\Lambda\circ\cdots\circ L_1(\mathbf{x}). \]
Nodes in $L_\Lambda$ are called \emph{output nodes}. We say that each output node $n_j^\Lambda$ in a neural network $N$ computes a function for which the value $N(\mathbf{x})_j$ is given in function of the input values $\mathbf{x}\in\mathbb{R}^{|L_0|}$.

We might restrict the input values of a neural network to some set $R\subseteq\mathbb{R}$. In this work, we focus on \emph{ReLU--TId neural networks}: they accept input values from $[0,1]$ and have as activation functions the \emph{rectified linear unit} $\rho_i=\mathrm{ReLU}:\mathbb{R}\to\mathbb{R}$, given by $\mathrm{ReLU}(x)=\max(0,x)$, for $i\in\{1,\ldots,\Lambda-1\}$, and the \emph{truncated identity function} $\rho_\Lambda=\mathrm{TId}:\mathbb{R}\to\mathbb{R}$, given by $\mathrm{TId}(x) = \max(0,\min(1,x))$. Such activation functions may be given by piecewise linear definitions as follows:
\begin{align}\label{eq:activation-functions}
	\mathrm{ReLU}(x) = \left\{
	\begin{array}{ll}
		0, & x < 0 \\
		x, & x \geq 0
	\end{array} \right.
	\hspace{80pt}
	\mathrm{TId}(x) = \left\{
	\begin{array}{ll}
		0, & x < 0 \\
		x, & 0 \leq x \leq 1 \\
		1, & x > 1
	\end{array} \right.
\end{align}

\begin{example}\label{ex:example}
	Let $E$ be a ReLU--TId neural network with $\mathcal{L}_E = \{ L_0, L_1, L_2 \}$, where $L_0 = \{ n_1^0, n_2^0 \}$, $L_1 = \{ n_1^1, n_2^1 \}$ and $L_2 = \{ n_1^2 \}$. The graph of $E$ depicted in Figure~\ref{fig:example} highlights the input values $x_1$ and $x_2$ in layer $L_0$ and the functions $f_1^1,f_2^1,f_1^2:\mathbb{R}^2\to\mathbb{R}$ in layers $L_1$ and $L_2$, which are given by:
	\begin{itemize}
		\item $f_1^1(x_1,x_2) = \frac{4}{3}x_1 - x_2$;
		\item $f_2^1(x_1,x_2) = x_1 - x_2 + \frac{1}{2}$;
		\item $f_1^2(x_1,x_2) = x_1 + x_2 + \frac{1}{2}$.
	\end{itemize}
	For a tuple of inputs $\mathbf{e}=\tuple{\frac{1}{8},\frac{1}{2}}$ to $E$, we have $L_1(\mathbf{e})=\tuple{\mathrm{ReLU}(-\frac{1}{3}),\mathrm{ReLU}(\frac{1}{8})}=\tuple{0,\frac{1}{8}}$ and, thus, $E(\mathbf{e})=L_2(L_1(\mathbf{e}))=\mathrm{TId}(\frac{5}{8})=\frac{5}{8}$.
\end{example}

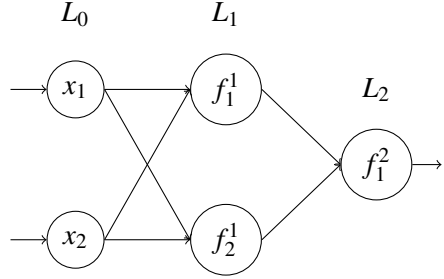
\begin{figure}
	\centering
	\begin{tikzpicture}[
		node/.style={circle, draw},
		clear/.style={draw=none, fill=none},
		]
		
		\node[clear]		(empty)							{};
		\node[node]			(input)		[right of=empty]	{$x_1$};
		\node[clear]        (cap)       [above of=input]    {$L_0$};
		\node[clear]		(empty2)	[right of=input]	{};
		\node[node]			(inner)		[right of=empty2]	{$f_1^1$};
		\node[clear]        (cap2)      [above of=inner]    {$L_1$};
		\node[clear]		(empty3)	[right of=inner]	{};
		\node[clear]		(empty4)	[below of=empty]	{};
		\node[clear]		(empty5)	[right of=empty4]	{};
		\node[clear]		(empty6)	[right of=empty5]	{};
		\node[clear]		(empty7)	[right of=empty6]	{};
		\node[clear]		(empty8)	[right of=empty7]	{};
		\node[node]			(output)	[right of=empty8]	{$f_1^2$};
		\node[clear]        (cap3)      [above of=output]   {$L_2$};
		\node[clear]		(empty-out)	[right of=output]	{};
		\node[clear]		(empty9)	[below of=empty4]	{};
		\node[node]			(input2)	[right of=empty9]	{$x_2$};
		\node[clear]		(empty10)	[right of=input2]	{};
		\node[node]			(inner2)	[right of=empty10]	{$f_2^1$};
		\node[clear]		(empty11)	[right of=inner2]	{};
		
		\draw[->] (empty.east)  -- (input.west);
		\draw[->] (input.east) -- (inner.west);
		\draw[->] (input.east) -- (inner2.west);
		\draw[->] (inner.east) -- (output.west);
		\draw[->] (empty9.east)  -- (input2.west);
		\draw[->] (input2.east) -- (inner.west);
		\draw[->] (input2.east) -- (inner2.west);
		\draw[->] (inner2.east) -- (output.west);
		\draw[->] (output.east) -- (empty-out.west);
	\end{tikzpicture}
	\caption{Graph of the ReLU--TId neural network $E$}
	\label{fig:example}
\end{figure}

Before introducing another type of neural network, let us define a \emph{rational McNaughton function} $f:[0,1]^n\to [0,1]$, which is a function that satisfies the following conditions:
\begin{itemize}
	\item $f$ is continuous with respect to the usual topology of $[0,1]$ as an interval of the real number line;
	\item There are linear polynomials $p_1,\ldots,p_m$ over $[0,1]^n$ with rational coefficients such that, for each point $\mathbf{x}\in [0,1]^n$, there is an index $i\in\{1,\ldots,m\}$ with $f(\mathbf{x})=p_i(\mathbf{x})$. Polynomials $p_1,\ldots,p_m$ are the \emph{linear pieces} of $f$.
\end{itemize}
A neural network whose $\nu=|L_\Lambda|$ output nodes exactly compute rational McNaughton functions in function of its input nodes is called a \emph{$\nu$-rational McNaughton neural network ($\nu$-RMcN\textsuperscript{3})}; a $1$-RMcN\textsuperscript{3} is also called \emph{rational McNaughton neural network (RMcN\textsuperscript{3})}.

A possibility to represent rational McNaughton functions (consequently, $\nu$-rational McNaughton neural networks) is through the \emph{regional formats} discussed in the following. Let $\inte{\Omega}$ denote the topological interior of $\Omega\subseteq [0,1]^n$ and say that, given functions $f,g:[0,1]^n\to [0,1]$, $f$ is \emph{above} $g$ over the set $\Omega$ if $f(\mathbf{x})\geq g(\mathbf{x})$, for all $\mathbf{x}\in\Omega$. A given rational McNaughton function $f:[0,1]^n\to [0,1]$ is said to be encoded in the \emph{closed regional format} if it is given by $m$ (not necessarily distinct) linear pieces
\begin{align}
	p_i(\mathbf{x})= \gamma_{i0} + \gamma_{i1}x_1 + \cdots + \gamma_{in}x_n,
\end{align}
where $\mathbf{x}=\tuple{x_1,\ldots,x_n}\in [0,1]^n$, $\gamma_{ij}\in\mathbb{Q}$ and $i\in\{1,\ldots,m\}$, such that each $p_i$ is identical to $f$ over a polyhedron $\Omega_i\subseteq [0,1]^n$, called \emph{region}. These regions are determined by the finite intersection of half-spaces given by linear inequalities\footnote{We occasionally abuse notation by using the same symbol to refer both to a set of inequalities and to the polyhedron it determines.} as
\begin{align}\label{eq:region}
	\Omega_i = \Big\{ \mathbf{x}\in [0,1]^n ~\Big|~ \omega_{j0} + \omega_{j1} x_1 + \cdots + \omega_{jn} x_n \geq 0,~ j\in\{1,\ldots,\lambda_{\Omega_i}\} \Big\}
\end{align}
and such setting of linear pieces and regions satisfy the following properties:
\begin{itemize}
	\item $\bigcup_{i=1}^m\Omega_i = [0,1]^n$;
	\item $\inte{\Omega_{i'}}\cap\inte{\Omega_{i''}}=\emptyset$, for $i'\neq i''$; and
	\item The \emph{lattice property}: for $i\neq j$, there is $k$ such that linear piece $p_i$ is above linear piece $p_k$ over region $\Omega_i$ and linear piece $p_k$ is above linear piece $p_j$ over region $\Omega_j$.
\end{itemize}
Such an encoding is called \emph{closed} because regions $\Omega_i$ are closed sets in the topological sense. As regions are given by such polyhedra described by \eqref{eq:region}, there is a polynomial procedure to establish whether a linear piece $p$ is above $q$ over region $\Omega$: find the minimum value $m$ of $p-q$ over $\Omega$, which is a linear program and may be solved in polynomial time \cite{BT1997}; then, if $m\geq 0$, $p$ is above $q$ over $\Omega$, otherwise, it is not.

The lattice property yields the possibility to represent rational McNaughton functions and $\nu$-RMcN\textsuperscript{3}s by \emph{lattice representations}---i.e., based on operations of maximum and minimum over functions---as follows. First, let $f_{\Omega_j}:[0,1]^n\to [0,1]$ be the function given by
\[ f_{\Omega_j}(\mathbf{x}) = \min\Big\{ p_k(\mathbf{x}) ~\Big|~ \text{$p_k$ is above $p_j$ over $\Omega_j$} \Big\}. \]
Note that $f_{\Omega_j}(\mathbf{x})\leq f(\mathbf{x})$, for all $\mathbf{x}\in [0,1]^n$, which is obvious for $\mathbf{x}\in\Omega_j$ and follows from the lattice property for $\mathbf{x}\in\Omega_i$ where $i\neq j$. In this way, we have that
\[ f(\mathbf{x}) = \max\Big\{ f_{\Omega_j}(\mathbf{x}) ~\Big|~ j\in\{1,\ldots,m\} \Big\}. \]
In \cite{PF2023b}, we find an example of the one-variable rational McNaughton function $f:[0,1]\to [0,1]$, whose graph is depicted in Figure \ref{fig:example-lattice}. Function $f$ has a lattice representation given by
\[ f(x) = \max\Big\{~ f_{\Omega_1}(x),~ f_{\Omega_2}(x),~ f_{\Omega_3}(x),~ f_{\Omega_4}(x)~ \Big\}, \]
where $f_{\Omega_1}(x)=\min\{p_1(x),p_3(x)\}$, $f_{\Omega_2}(x)=f_{\Omega_3}(x)=\min\{p_2(x),p_3(x)\}$ and $f_{\Omega_4}(x)=\min\{p_2(x),p_4(x)\}$. Note that the lattice encoding just introduced may be employed in representing piecewise linear functions in general, not only rational McNaughton functions.

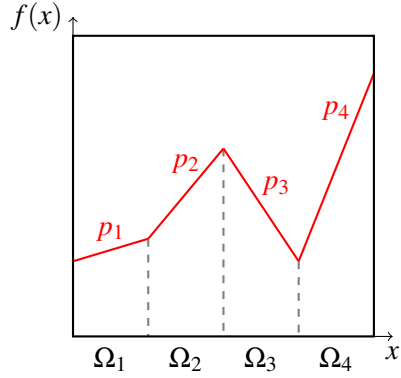
\begin{figure}
	\centering
	\begin{tikzpicture}
		\draw[gray,dashed,thick] (1,1.3)--(1,0);
		\draw[red,thick] (0,1)--(1,1.3)
		node at (0.5,1.4) {$p_1$};
		\draw[gray,dashed,thick] (2,2.5)--(2,0);
		\draw[red,thick] (1,1.3)--(2,2.5)
		node at (1.5,2.3) {$p_2$};
		\draw[gray,dashed,thick] (3,1)--(3,0);
		\draw[red,thick] (2,2.5)--(3,1)
		node at (2.7,2) {$p_3$};
		\draw[red,thick] (3,1)--(4,3.5)
		node at (3.5,3) {$p_4$};
		\draw[thick] (0,0) rectangle (4,4);
		\draw[->] (0,0)--(0,4+1/4)
		node[left] {$f(x)$};
		\draw[->] (0,0)--(4+1/4,0)
		node at (0.5,-0.3) {$\Omega_1$}
		node at (1.5,-0.3) {$\Omega_2$}
		node at (2.5,-0.3) {$\Omega_3$}
		node at (3.5,-0.3) {$\Omega_4$}
		node[below] {$x$};
	\end{tikzpicture}
	\caption{One-variable piecewise linear function}
	\label{fig:example-lattice}
\end{figure}

When a rational McNaughton function is given in an encoding that almost completely agrees with the closed regional format, with the sole exception that there is no guarantee that such encoding satisfies the lattice property (although it may still satisfy), we say that it is in the \emph{pre-closed regional format}.

Unfortunately, lack of lattice property might entail the failure of lattice representation as in the following example taken from \cite{PF2022}. The rational McNaughton function $f_E$ with graph in Figure \ref{fig:graph-ce} may have an encoding based on regions in Figure \ref{fig:regions-ce}; a linear piece $p_i$ is associated to each region $\Omega_i$. The dotted line in Figure \ref{fig:regions-ce} is the projection over $[0,1]^2$ of where $p_3$ intercepts $p_5$; note that such line passes through the interior of both $\Omega_3$ and $\Omega_5$. There is no linear piece $p_k$ such that $p_3$ is above $p_k$ over $\Omega_3$ and $p_k$ is above $p_5$ over $\Omega_5$. So an encoding for $f_E$ based on regions $\Omega_1$--$\Omega_5$ may be at most encoded in pre-closed regional format.

Now, there is $\mathbf{x}_0\in\Omega_3^\circ$ such that $p_5(\mathbf{x}_0)>p_3(\mathbf{x}_0)$. Therefore, $f_{\Omega_5}(\mathbf{x}_0) = p_5(\mathbf{x}_0) > p_3(\mathbf{x}_0) = \min\{ p_1(\mathbf{x}_0), p_3(\mathbf{x}_0) \} = f_{\Omega_3}(\mathbf{x}_0)$, yielding that $\max \{f_{\Omega_j}(\mathbf{x}_0)\} > f_{\Omega_3}(\mathbf{x}_0)$, which eliminates the possibility of lattice representation. Such an issue may be circumvented by splitting region $\Omega_5$, according to the dotted line in Figure \ref{fig:regions-ce}, in regions $\Omega'_5 = \Omega_5\cap\{p_5-p_3\geq0\}$ and $\Omega''_5 = \Omega_5\cap\{p_5-p_3\leq0\}$. In general, repeatedly splitting a region according to projections of linear pieces intersections eventually achieves closed regional format \cite[Theorem 7]{PF2022}.

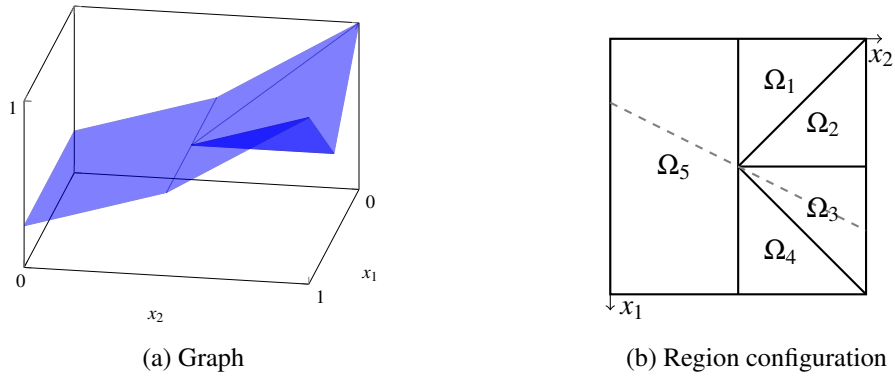
\begin{figure}
	\centering
	\begin{subfigure}{0.45\textwidth}
		\centering
		\begin{tikzpicture}[scale=.65]
			\begin{axis}[view/h=100, xlabel=$x_1$, ylabel=$x_2$, xmin=0, xmax=1, ymin=0, ymax=1, zmin=0, zmax=1, xtick={0,1}, ytick={0}, ztick={1}]
				\addplot3[opacity=0.5,table/row sep=\\,patch,patch type=polygon,vertex count=4,patch table with point meta={
					0 1 3 2 0\\
					2 4 5 2\\
					5 4 7 5\\
					6 7 4 6\\
					3 4 6 3\\
				}]
				table {
					x y z\\
					0 0 0.25\\
					1 0 0.25\\
					0 0.5 0.5\\
					1 0.5 0.5\\
					0.5 0.5 0.5\\
					0 1 1\\
					1 1 1\\
					0.5 1 0.5\\
				};
			\end{axis}
		\end{tikzpicture}
		\caption{Graph}
		\label{fig:graph-ce}
	\end{subfigure}
	~
	\begin{subfigure}{0.45\textwidth}
		\centering
		\begin{tikzpicture}[scale=.85]
			\draw[thick] (0,0) rectangle (4,4);
			\draw[->] (0,4)--(0,-1/4)
			node[right] {$x_1$};
			\draw[->] (0,4)--(4+1/4,4)
			node[below] {$x_2$};
			\draw[thick] (2,0)--(2,4);
			\draw[thick] (4,0)--(2,2);
			\draw[thick] (4,2)--(2,2);
			\draw[thick] (2,2)--(4,4);
			\draw[gray,dashed,thick] (0,3)--(4,1);
			\node at (8/3,2/3) {$\Omega_4$};
			\node at (10/3,4/3) {$\Omega_3$};
			\node at (10/3,8/3) {$\Omega_2$};
			\node at (8/3,10/3) {$\Omega_1$};
			\node at (1,2) {$\Omega_5$};
		\end{tikzpicture}
		\caption{Region configuration}
		\label{fig:regions-ce}
	\end{subfigure}
	\caption{Function encoded in the pre-closed regional format}
	\label{fig:graph-regions-ce}
\end{figure}

We may also represent rational McNaughton functions and $\nu$-RMcN\textsuperscript{3}s in logical systems. For that, let us introduce the \Luka\ infinitely-valued logic (\Linf). The basic language $\mathcal{L}$ of \Linf\ comprehends formulas freely generated from a countable set of propositional variables $\mathbb{P}$, a disjunction operator $\oplus$ and a negation operator $\lnot$. A \emph{valuation} is a function $v: \mathcal{L} \to [0,1]$, such that, for $\phi,\psi\in\mathcal{L}$:
\begin{align}
	v(\phi \oplus \psi) &= \min( 1, v(\phi) + v(\psi) );\label{eq:oplusval} \\
	v( \lnot \phi) &= 1-v(\phi).\label{eq:negval}
\end{align}
From disjunction and negation we derive the following operators:
\begin{align*}
	\textrm{Conjunction: } & \phi \odot \psi =_\mathrm{def} \lnot(\lnot\phi \oplus \lnot\psi) & v(\phi \odot \psi) &= \max(0, v(\phi)+v(\psi) - 1)	\\
	\textrm{Implication: } & \phi \to \psi =_\mathrm{def} \lnot\phi \oplus \psi & v(\phi \to \psi) &= \min(1, 1-v(\phi)+v(\psi))	\\
	\textrm{Maximum: } & \phi \lor \psi =_\mathrm{def} \lnot(\lnot\phi \oplus \psi) \oplus \psi & v(\phi \lor \psi) &= \max(v(\phi),v(\psi))	\\
	\textrm{Minimum: } & \phi \land \psi =_\mathrm{def} \lnot(\lnot\phi \lor \lnot\psi) & v(\phi \land \psi) &= \min(v(\phi), v(\psi))	\\
	\textrm{Bi-implication: } & \phi \leftrightarrow \psi =_\mathrm{def} (\phi \to \psi) \land (\psi \to \phi) & v(\phi \leftrightarrow \psi) &= 1-|v(\phi)-v(\psi)|
\end{align*}
Note that, as lattice operations are expressed in \Linf\ by the minimum and maximum operators, piecewise linear functions might have a lattice representation in \Linf\ as far as their linear pieces are representable in this system. Indeed, the formulas of \Linf\ represent all the \emph{McNaughton functions}, which are rational McNaughton functions constrained to allow only integer coefficients in their linear pieces \cite{McN1951,Mun1994}.

Unfortunately, \Linf\ cannot express rational McNaughton functions. For that, one possible path is to extend the language of \Linf, which is done, for instance, by \cite{Ger2001}. Another possibility is to implicitly represent such functions in plain \Linf\ using the technique of \emph{representation modulo satisfiability}, which we introduce in the following \cite{FP2020,PF2020,PF2022}.

Let us denote the \Linf-\emph{semantics}, that is the set of all valuations, by $\mathbf{Val}$. Let us also denote by $\mathbf{Val}_\Phi$ the set of valuations $v\in\mathbf{Val}$ that satisfy a set of formulas $\Phi$; we call such a restricted set of valuations a \emph{semantics modulo satisfiability}. Given a rational McNaughton function $f:[0,1]^n\to[0,1]$, a formula $\phi_f$ and a set of formulas $\Phi_f$, we say that \emph{$\phi_f$ represents $f$ modulo $\Phi_f$-satisfiable} or that the pair $\langle\phi_f,\Phi_f\rangle$ \emph{represents $f$ (in the system \Linf-MODSAT)} if, for distinguished propositional variables $X_1,\ldots,X_n\in\mathbb{P}$:
\begin{itemize}
	\item For all $\langle x_1,\ldots,x_n\rangle\in[0,1]^n$, there exists some valuation $v\in\mathbf{Val}_{\Phi_f}$, such that $v(X_i)=x_i$, for $i=1,\ldots,n$;
	\item For all valuations $v,v'\in\mathbf{Val}_{\Phi_f}$ such that $v(X_i)=v'(X_i)$, for $i=1,\ldots,n$, we have $v(\phi_f)=v'(\phi_f)$; and
	\item $f(v(X_1),\ldots,v(X_n)) = v(\phi_f)$, for all $v\in\mathbf{Val}_{\Phi_f}$.
\end{itemize}

As an example, any constant function that takes value $\frac{1}{b}$, with $b\in\mathbb{N}^*$, may be represented by the pair
\begin{align}\label{eq:constant-representation}
	\tuple{\phi,\Phi} = \Big\langle ~~ Z_\frac{1}{b}, ~~ \Big\{ Z_\frac{1}{b}\leftrightarrow\neg(b-1)Z_\frac{1}{b} \Big\} ~~ \Big\rangle,
\end{align}
where formula $\phi$ is only the propositional variable $Z_{\sfrac{1}{b}}$ and set $\Phi$ is a singleton comprehending formula $Z_{\sfrac{1}{b}}\leftrightarrow\neg(b-1)Z_{\sfrac{1}{b}}$, which we denote by $\phi_{\sfrac{1}{b}}$. In fact, for any valuation $v\in\mathbf{Val}_{\phi_{\sfrac{1}{b}}}\neq\emptyset$, we have that $v(Z_{\sfrac{1}{b}})=\frac{1}{b}$. Also, functions that take constant value $\frac{a}{b}$, with $a\in\mathbb{N}$, may be represented by the pair $\tuple{a\phi,\Phi}$.

Any rational McNaughton function may be represented in \Linf-MODSAT. Moreover, there is a polynomial algorithm for the translation from a rational McNaughton function in closed regional format to its representation in such system \cite{PF2020,PF2022}.

\section{Neural Networks into Pre-Closed Regional Format}
\label{sec:nn2pwl}
Given a ReLU--TId neural network $N$ for which $\mathcal{L}_N = \{ L_0,\ldots,L_\Lambda \}$, we provide an algorithm to translate it into a tuple $\Xi_N = \tuple{\Xi_1,\ldots,\Xi_{|L_\Lambda|}}$, where each $\Xi_k$, $k\in\{1,\ldots,|L_\Lambda|\}$, is the codification for a rational McNaughton function in pre-closed regional format, which we will show to be the function computed by $N$ through the path to its $k$-th output node. Each $\Xi_k$ is a set of pairs $\tuple{p,\Omega}$, where $p$ is a linear piece and $\Omega$ is its associated region.

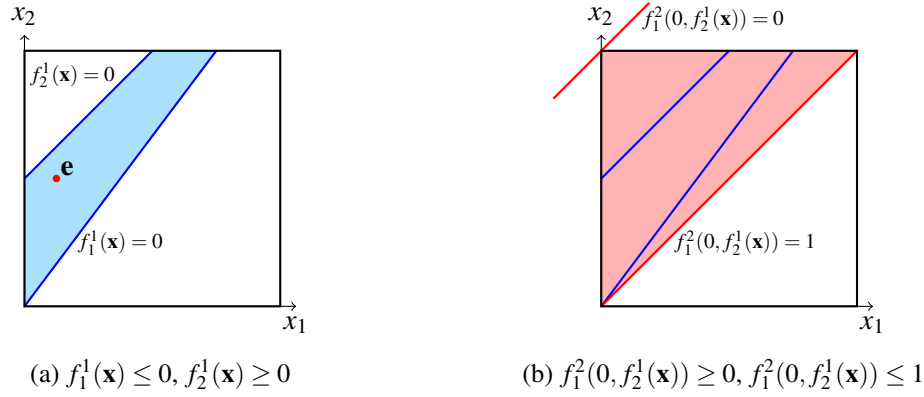
\begin{figure}
	\centering
	\begin{subfigure}{0.45\textwidth}
		\centering
		\begin{tikzpicture}[scale=.85]
			\draw[fill=cyan!30,ultra thin] (0,0)--(3,4)--(2,4)--(0,2)--cycle;
			\draw[->] (0,0)--(0,4+1/4)
			node[above] {$x_2$};
			\draw[->] (0,0)--(4+1/4,0)
			node[below] {$x_1$};
			\draw[blue,thick] (0,0)--(3,4);
			\node at (1+1/2,1) {\scriptsize $f_1^1(\mathbf{x})=0$};
			\draw[blue,thick] (0,2)--(2,4);
			\node at (1/2+1/4,3+1/2+1/8) {\scriptsize $f_2^1(\mathbf{x})=0$};
			\node at (1/2,2) [red,circle,fill,inner sep=1pt] {};
			\node at (1/2+1/6,2+1/6) {$\mathbf{e}$};
			\draw[thick] (0,0) rectangle (4,4);
		\end{tikzpicture}
		\caption{$f_1^1(\mathbf{x})\leq 0$, $f_2^1(\mathbf{x})\geq 0$}
		\label{fig:comp-inner-node}
	\end{subfigure}
	~
	\begin{subfigure}{0.45\textwidth}
		\centering
		\begin{tikzpicture}[scale=.85]
			\draw[fill=red!30,ultra thin] (0,0)--(4,4)--(0,4)--cycle;
			\draw[->] (0,0)--(0,4+1/4)
			node[above] {$x_2$};
			\draw[->] (0,0)--(4+1/4,0)
			node[below] {$x_1$};
			\draw[blue,thick] (0,0)--(3,4);
			\draw[blue,thick] (0,2)--(2,4);
			\draw[red,thick] (0,0)--(4,4);
			\node at (2+1/4,1) {\scriptsize $f_1^2(0,f_2^1(\mathbf{x}))=1$};
			\draw[red,thick] (-1/2-1/4,-1/2-1/4+4)--(1/2+1/4,1/2+1/4+4);
			\node at (1+3/4,4+1/2) {\scriptsize $f_1^2(0,f_2^1(\mathbf{x}))=0$};
			\draw[thick] (0,0) rectangle (4,4);
		\end{tikzpicture}
		\caption{$f_1^2(0,f_2^1(\mathbf{x}))\geq 0$, $f_1^2(0,f_2^1(\mathbf{x}))\leq 1$}
		\label{fig:comp-last-node}
	\end{subfigure}
	\caption{Determining a region for neural network $E$}
\end{figure}

Let us begin by analyzing the computation that takes place in each node of $N$. Given a tuple of input values $\mathbf{v}\in\mathbb{R}^{|L_{i-1}|}$ to node $n_j^i$ of $N$, where $i\in\{1,\ldots,\Lambda-1\}$ and so $\rho_i=\mathrm{ReLU}$, the computation proceeds in two steps: first, the linear function $f_j^i(\mathbf{x})$ is evaluated for $\mathbf{x}=\mathbf{v}$; second, the activation function $\mathrm{ReLU}(x)$ is evaluated for $x=f_j^i(\mathbf{v})$. In the second step, one from the two possible values highlighted in the piecewise definition of $\mathrm{ReLU}$ in \eqref{eq:activation-functions} is chosen as the output of $n_j^i$. Such choice depends on the position of point $\mathbf{v}\in\mathbb{R}^{|L_{i-1}|}$ in relation to the hyperplane given by equation $f_j^i(\mathbf{x}) = 0$, since it may be a point lying either in the half-space given by $f_j^i(\mathbf{x}) \leq 0$ or in the half-space given by $f_j^i(\mathbf{x}) \geq 0$. For instance, for $\mathbf{e}=\tuple{\frac{1}{8},\frac{1}{2}}$ as input to the node $n_1^1$ of $E$ in Example \ref{ex:example}, the fact that $f_1^1(\mathbf{e}) \leq 0$ indicates that $\mathbf{e}$ lies in one of the half-spaces determined by $f_1^1(\mathbf{x}) = 0$ where, according to $\rho_1=\mathrm{ReLU}$, $n_1^1$ outputs $0$. On the other hand, if $\mathbf{e}$ is given as input to the node $n_2^1$, as $f_2^1(\mathbf{e}) \geq 0$, $\mathbf{e}$ lies in the half-space determined by $f_2^1(\mathbf{x})=0$ where, according to $\rho_1=\mathrm{ReLU}$, $n_2^1$ outputs $f_2^1(\mathbf{e})=\frac{1}{8}$. Figure \ref{fig:comp-inner-node} depicts the position of $\mathbf{e}$ in relation to these hyperplanes.

Similarly, in case $i=\Lambda$, we have that $\rho_i=\mathrm{TId}$ and, then, one from three possible values is chosen for $\mathrm{TId}(f_j^i(\mathbf{v}))$ depending on the position of $\mathbf{v}\in\mathbb{R}^{|L_{\Lambda-1}|}$ in relation to the hyperplanes $f_j^\Lambda(\mathbf{x})=0$ and $f_j^\Lambda(\mathbf{x})=1$. It may be a point lying either in the half-space given by $f_j^\Lambda(\mathbf{x})\leq 0$, or in the half-space given by $f_j^\Lambda(\mathbf{x})\geq 1$ or in the intersection of half-spaces $f_j^\Lambda(\mathbf{x})\geq 0$ and $f_j^\Lambda(\mathbf{x})\leq 1$. Again, for $\mathbf{e}=\tuple{\frac{1}{8},\frac{1}{2}}$ as input to the neural network $E$ in Example \ref{ex:example}, we have $\mathbf{e}_1=L_1(\mathbf{e})=\tuple{0,\frac{1}{8}}$. Since $f_1^2(\mathbf{e}_1)\geq 0$ and $f_1^2(\mathbf{e}_1)\leq 1$, $\mathbf{e}_1$ lies in a position relative to $f_1^2(\mathbf{x})=0$ and $f_1^2(\mathbf{x})=1$ where, according to $\rho_2=\mathrm{TId}$, $n_1^2$ outputs $f_1^2(\mathbf{e}_1)=\frac{5}{8}$.

Now, still considering Example \ref{ex:example}, let $\mathbf{x}\in [0,1]^{|L_0|}$ be any point that, as $\mathbf{e}=\tuple{\frac{1}{8},\frac{1}{2}}$, satisfies both inequalities
\begin{align}\label{eq:region-eq1}
	f^1_1(\mathbf{x})\leq 0 ~~~~\text{and}~~~~ f^1_2(\mathbf{x})\geq 0. 
\end{align}
Then, to $\mathbf{x}_1=L_1(\mathbf{x})$ satisfy inequalities $f^2_1(\mathbf{x}_1)\geq 0$ and $f^2_1(\mathbf{x}_1)\leq 1$ is equivalent to $\mathbf{x}$ satisfy inequalities
\begin{align}\label{eq:region-eq2}
	f^2_1(0,f^1_2(\mathbf{x}))\geq 0 ~~~~\text{and}~~~~ f^2_1(0,f^1_2(\mathbf{x}))\leq 1.
\end{align}
As the former inequalities \eqref{eq:region-eq1}, the latter inequalities \eqref{eq:region-eq2} are also linear over tuples from $[0,1]^{|L_0|}$ of input values to the neural network $E$. Moreover, for $\mathbf{x}\in[0,1]^{|L_0|}$ satisfying inequalities \eqref{eq:region-eq1} and \eqref{eq:region-eq2}, we have that $E(\mathbf{x})=f^2_1(0,f^1_2(\mathbf{x}))$. In this way, we have just devised a region and its associated linear piece for the rational McNaughton function computed by neural network $E$.

Generalizing the observations above, the idea behind the base algorithm for building $\Xi_k$, for $k\in\{1,\ldots, |L_\Lambda|\}$, is to compute each pair $\tuple{p,\Omega}\in\Xi_k$ beginning by: associating a symbol between $\leq$ and $\geq$ to each node $n_j^i$, for $i<\Lambda$, alluding to one of the two possible positions of an input to $n_j^i$ in relation to the hyperplane $f_j^i(\mathbf{x})=0$---i.e., lying in the half-space $f_j^i(\mathbf{x})\leq 0$ or in the half-space $f_j^i(\mathbf{x})\geq 0$\mbox{---;} and associating a symbol among $\leq$, $\geq$ and $\lessgtr$ to the node $n_k^\Lambda$ alluding to one of the three possible positions of an input to $n_k^\Lambda$ in relation to the hyperplanes $f_k^\Lambda(\mathbf{x})=0$ and $f_k^\Lambda(\mathbf{x})=1$---i.e., lying in the half-space $f_k^\Lambda(\mathbf{x})\leq 0$, or in the half-space $f_k^\Lambda(\mathbf{x})\geq 1$ or in the intersection of the half-spaces $f_k^\Lambda(\mathbf{x})\geq 0$ and $f_k^\Lambda(\mathbf{x})\leq 1$. These associations of symbols to all the nodes in layers $L_1,\ldots,L_\Lambda$ determine a \emph{configuration of symbols}. Then, the algorithm proceeds by defining $\Omega\subseteq[0,1]^{|L_0|}$ as an intersection of half-spaces based on such configuration of symbols and establishing a linear expression for $p$ such that $N(\mathbf{x})_k=p(\mathbf{x})$, for $\mathbf{x}\in\Omega$.

\begin{example}\label{ex:example-region}
	For the neural network $E$ in Example \ref{ex:example}, in a configuration of symbols where we associate $\leq$ to $n_1^1$ and $\geq$ to $n_2^1$, the consequent region $\Omega$ should comprehend the inequalities $\frac{4}{3}x_1-x_2\leq 0$ and $x_1-x_2+\frac{1}{2}\geq 0$ (shaded area in Figure \ref{fig:comp-inner-node}). For an input $\mathbf{x}\in[0,1]^2$ that satisfies these inequalities, we have the outputs $n_1^1(\mathbf{x}) = 0$ and $n_2^1(\mathbf{x}) = f_2^1(\mathbf{x})$ which, composed with $f_1^\Lambda$, gives us the expression $x_1-x_2+1$. In this way, in case we complete the configuration of symbols by associating $\leq$ to $n_1^\Lambda$, $\Omega$ should comprehend the inequality $x_1-x_2+1\leq 0$ and $p$ should be given by $p(x_1,x_2) = 0$. In case we associate $\geq$ to $n_1^\Lambda$, $\Omega$ should comprehend the inequality $x_1-x_2+1\geq 1$ and $p$ should be given by $p(x_1,x_2) = 1$. In the last case, if we associate $\lessgtr$ to $n_1^\Lambda$, $\Omega$ should comprehend both inequalities $x_1-x_2+1\leq 1$ and $x_1-x_2+1\geq 0$ (shaded area in Figure \ref{fig:comp-last-node}) and $p$ should be given by $p(x_1,x_2)=x_1-x_2+1$. Note that the last case is the only one where region $\Omega$ would be non-empty.
\end{example}

In order to build the entire representative tuple $\Xi_N = \tuple{\Xi_1,\ldots,\Xi_{|L_\Lambda|}}$, the algorithm needs to compute all the pairs $\tuple{p,\Omega}$, each one associated to a different configuration of symbols, for all possible configurations of symbols. Thus, the entire computation of $\Xi_N$ ends up with $2^{|L_1|} \times \cdots \times 2^{|L_{\Lambda-1}|} \times 3|L_\Lambda|$ pairs $\tuple{p,\Omega}$. Later, we introduce methods meant to be combined with the base algorithm that might circumvent such high complexity.

For establishing the base translation algorithm, we first fix some notation. Let $\kappa_0^n$ and $\kappa_1^n$ be the constant linear functions with domain $[0,1]^n$ and ranges equal to $\{0\}$ and $\{1\}$, respectively; and let $\chi_n:\{\leq,\geq\}\to\{\kappa_0^n,\kappa_1^n\}$ be the functions given by $\chi_n(\leq)=\kappa_0^n$ and $\chi_n(\geq)=\kappa_1^n$. Also, let $\pi_n^m:[0,1]^m\to\mathbb{R}$ be the projection functions given by $\pi_n^m(x_1,\ldots,x_m) = x_n$, for $m\in\mathbb{N}$ and $1\leq n\leq m$. The base translation algorithm is split into Algorithms \ref{alg:nn2pwl} and \ref{alg:nn2pwl-inner}.

Algorithm~\ref{alg:nn2pwl} treats the tuple $\Xi_N=\tuple{\Xi_1,\ldots,\Xi_{|L_\Lambda|}}$ as a variable to be updated as it runs; thus, it first sets each of the $\Xi_1,\ldots,\Xi_{|L_\Lambda|}$ to the empty set as their initial values (lines \ref{alg:nn2pwl-initial-empty1} and \ref{alg:nn2pwl-initial-empty2}). Then, it defines $\Omega_{[0,1]}$ as a set of inequalities common to all regions (line \ref{alg:nn2pwl-generic-region}) and $\pi$ as a tuple of projection functions (line \ref{alg:nn2pwl-projection-tuple}), which will be suitable for compositions with functions $f_j^1$ related to the first layer $L_1$. It proceeds by calling the recursive routine $\textsc{nn2pwl-r}(\Xi_N \parallel L_1,\Omega,\pi)$ (line \ref{alg:nn2pwl-recursive-call}), where $\Xi_N$ is an argument \textbf{passed by reference}, which means that whenever \textsc{nn2pwl-r} modifies the value of $\Xi_N$, it will also be modified in the scope of the calling function. Finally, Algorithm \ref{alg:nn2pwl} returns $\Xi_N$ with its final value (line \ref{alg:nn2pwl-return}).

Algorithm \ref{alg:nn2pwl-inner} describes the recursive routine \textsc{nn2pwl-r} that has as inputs a tuple $\Xi_N=\tuple{\Xi_1,\ldots,\Xi_{|L_\Lambda|}}$ to be updated, a ReLU--TId neural network $N$ with a distinguished layer $L_i\in\mathcal{L}_N$, a set of inequalities $\Omega$ and a tuple of functions $f=\tuple{f_1,\ldots,f_{|L_{i-1}|}}$. If $L_i\neq L_\Lambda$, for all possible association of symbols $\leq$ and $\geq$ to the nodes $n_1^i,\ldots,n_{|L_i|}^i$ summarized in the tuple of symbols $\bowtie=\tuple{\bowtie_1,\ldots,\bowtie_{|L_i|}}$, $\textsc{nn2pwl-r}(\Xi_N \parallel L_i,\Omega,f)$ proceeds by:
\begin{itemize}
	\item Computing $\Omega^i_{\bowtie}$ as $\Omega$ extended by the half-spaces $f_j^i\circ f(\mathbf{x}) \bowtie_j 0$, for $j\in\{1,\ldots,|L_i|\}$, where $f=\tuple{f_1,\ldots,f_{|L_{i-1}|}}$ is a tuple of linear functions such that $f(\mathbf{x})=\tuple{f_1(\mathbf{x}),\ldots,f_{|L_{i-1}|}(\mathbf{x})} = L_{i-1}\circ\cdots\circ L_1(\mathbf{x})$, for $\mathbf{x}\in\Omega$ (line \ref{alg:nn2pwl-inner-halfspaces1});
	\item Computing the tuple of linear functions $f^i_{\bowtie}$ that is identical to the output of $L_i\circ\cdots\circ L_1$ for inputs $\mathbf{x}\in\Omega^i_{\bowtie}$, with assistance of functions $\chi_{|L_0|}$ (line \ref{alg:nn2pwl-inner-functions1});
	\item And calling itself again by $\textsc{nn2pwl-r}(\Xi_N \parallel L_{i+1},\Omega^i_{\bowtie},f^i_{\bowtie})$ (line \ref{alg:nn2pwl-inner-recursive-call1}).
\end{itemize}
If $L_i=L_\Lambda$, for each of the output nodes $n_1^\Lambda,\ldots,n_{|L_\Lambda|}^\Lambda$, $\textsc{nn2pwl-r}(\Xi_N \parallel L_i,\Omega,f)$ proceeds by:
\begin{itemize}
	\item Computing $\Omega_{\leq}$, $\Omega_{\geq}$ and $\Omega_{\lessgtr}$ as $\Omega$ extended, respectively, by the half-space $f_k^\Lambda\circ f(\mathbf{x}) \leq 0$, the half-space $f_k^\Lambda\circ f(\mathbf{x}) \geq 1$ and the pair of half-spaces $f_k^\Lambda\circ f(\mathbf{x}) \geq 0$ and $f_k^\Lambda\circ f(\mathbf{x}) \leq 1$, where $f$ is a tuple of linear functions such that $f(\mathbf{x}) = L_{\Lambda-1}\circ\cdots\circ L_1(\mathbf{x})$, for $\mathbf{x}\in\Omega$ (lines \ref{alg:nn2pwl-inner-omega1}, \ref{alg:nn2pwl-inner-omega2} and \ref{alg:nn2pwl-inner-omega3});
	\item And rewriting $\Xi_k$ by adding the pairs $\tuple{\kappa_0^{|L_0|}, \Omega_{\leq}}$, $\tuple{f_k^\Lambda\circ f, \Omega_{\lessgtr}}$ and $\tuple{\kappa_1^{|L_0|}, \Omega_{\geq}}$ to it (lines \ref{alg:nn2pwl-inner-xi1}, \ref{alg:nn2pwl-inner-xi2} and \ref{alg:nn2pwl-inner-xi3}).
\end{itemize}

Let $\Omega$ be any region appearing in the output of the base algorithm; it is built in $\Lambda+1$ steps in a way that, in each step, new inequalities are added to a polyhedron (identified with a set of inequalities) until it becomes $\Omega$. The first step adds the inequalities that determine $[0,1]^{|L_0|}$ (Algorithm~\ref{alg:nn2pwl}, line \ref{alg:nn2pwl-generic-region}). The next $\Lambda-1$ steps, where the produced polyhedra are named $\Omega^i_{\bowtie}$ (Algorithm~\ref{alg:nn2pwl-inner}, line \ref{alg:nn2pwl-inner-halfspaces1}), are associated to layers $L_1,\ldots,L_{\Lambda-1}$ of $N$. The final step, where the produced region $\Omega$ is named either as $\Omega_{\leq}$, $\Omega_{\geq}$ or $\Omega_{\lessgtr}$ (Algorithm~\ref{alg:nn2pwl-inner}, line~\ref{alg:nn2pwl-inner-omega1}, \ref{alg:nn2pwl-inner-omega2} or \ref{alg:nn2pwl-inner-omega3}), is associated to layer $L_\Lambda$.

\begin{algorithm}
	\caption{\textsc{nn2pwl}: puts neural networks in the closed regional format}\label{alg:nn2pwl}
	
	\textbf{Input:} A ReLU--TId neural network $N$ for which $\mathcal{L}_N = \{ L_0,\ldots,L_\Lambda \}$.
	
	\textbf{Output:} A set $\Xi_N$ representing rational McNaughton functions computed by the output nodes of $N$.
	
	\begin{algorithmic}[1]
		\STATE $\Xi_1 := \emptyset$, \ldots, $\Xi_{|L_\Lambda|} := \emptyset$; \label{alg:nn2pwl-initial-empty1}
		\STATE $\Xi_N := \tuple{\Xi_1,\ldots,\Xi_{|L_\Lambda|}}$; \label{alg:nn2pwl-initial-empty2}
		\STATE $\Omega_{[0,1]} := \{ x_1 \geq 0, x_1 \leq 1, \ldots, x_{|L_0|}\geq 0, x_{|L_0|}\leq 1 \}$; \label{alg:nn2pwl-generic-region}
		\STATE $\pi := \tuple{\pi_1^{|L_0|},\ldots,\pi_{|L_0|}^{|L_0|}}$; \label{alg:nn2pwl-projection-tuple}
		\STATE $\textsc{nn2pwl-r}(\Xi_N \parallel L_1,\Omega_{[0,1]},\pi)$; \label{alg:nn2pwl-recursive-call}
		\RETURN $\Xi_N$; \label{alg:nn2pwl-return}
	\end{algorithmic}
\end{algorithm}

\begin{algorithm}
	\caption{\textsc{nn2pwl-r}: recursive routine called by \textsc{nn2pwl}}\label{alg:nn2pwl-inner}
	
	\textbf{Input:} A tuple $\Xi_N=\tuple{\Xi_1,\ldots,\Xi_{|L_\Lambda|}}$, a ReLU-TId neural network $N$, for which $\mathcal{L}_N = \{ L_0,\ldots,L_\Lambda \}$, with a distinguished layer $L_i\neq L_0$, a set of inequalities $\Omega$ and a tuple of linear functions $f = \tuple{f_1,\ldots,f_{|L_{i-1}|}}$.
	
	\begin{algorithmic}[1]
		\IF{$L_i\neq L_\Lambda$}\label{alg:nn2pwl-inner-firstIf}
		\FOR{$\bowtie \in \{ \leq,\geq \}^{|L_i|}$} \label{alg:nn2pwl-inner-For}
		\STATE $\Omega^i_{\bowtie} := \Omega \cup \{ f_j^i\circ f(\mathbf{x}) \bowtie_j 0 ~|~ j=1,\ldots,|L_i| \}$; \label{alg:nn2pwl-inner-halfspaces1}
		\STATE $f^i_{\bowtie} := \tuple{\chi_{|L_0|}(\bowtie_1) \cdot (f_1^i\circ f), \ldots, \chi_{|L_0|}(\bowtie_{|L_i|}) \cdot (f_{|L_i|}^i\circ f)}$; \label{alg:nn2pwl-inner-functions1}
		\STATE $\textsc{nn2pwl-r}(\Xi_N \parallel L_{i+1},\Omega^i_{\bowtie},f^i_{\bowtie})$; \label{alg:nn2pwl-inner-recursive-call1}
		\ENDFOR \label{alg:nn2pwl-inner-For-end}
		\ELSE
		\FOR {$k = 1,\ldots,|L_\Lambda|$}
		\STATE $\Omega_{\leq} := \Omega \cup \{ f_k^\Lambda\circ f(\mathbf{x}) \leq 0 \}$; \label{alg:nn2pwl-inner-omega1}
		\STATE $\Xi_k := \Xi_k \cup \{ \tuple{\kappa_0^{|L_0|}, \Omega_{\leq}} \}$; \label{alg:nn2pwl-inner-xi1}
		\STATE $\Omega_{\lessgtr} := \Omega \cup \{ f_k^\Lambda\circ f(\mathbf{x}) \geq 0,~ f_k^\Lambda\circ f(\mathbf{x}) \leq 1 \}$; \label{alg:nn2pwl-inner-omega2}
		\STATE $\Xi_k := \Xi_k \cup \{ \tuple{f_k^\Lambda\circ f, \Omega_{\lessgtr}} \}$; \label{alg:nn2pwl-inner-xi2}
		\STATE $\Omega_{\geq} := \Omega \cup \{ f_k^\Lambda\circ f(\mathbf{x}) \geq 1 \}$; \label{alg:nn2pwl-inner-omega3}
		\STATE $\Xi_k := \Xi_k \cup \{ \tuple{\kappa_1^{|L_0|}, \Omega_{\geq}} \}$; \label{alg:nn2pwl-inner-xi3}
		\ENDFOR
		\ENDIF
	\end{algorithmic}
\end{algorithm}

\begin{lemma}\label{thm:nn2pwl-omega}
	Let a ReLU--TId neural network $N$, for which $\mathcal{L}_N = \{ L_0,\ldots,L_\Lambda \}$, be given as input to Algorithm~\ref{alg:nn2pwl}. Then, Algorithm \ref{alg:nn2pwl} terminates and outputs a tuple $\Xi_N=\tuple{\Xi_1,\ldots,\Xi_{|L_\Lambda|}}$ where, for $k\in\{1,\ldots,|L_\Lambda|\}$, we have:
	\begin{itemize}
		\item $\bigcup_{\Omega \text{, for } \tuple{p,\Omega}\in\Xi_k} \Omega = [0,1]^{|L_0|}$;
		\item $\inte{\Omega'}\cap\inte{\Omega''}=\emptyset$, for distinct $\tuple{p',\Omega'},\tuple{p'',\Omega''}\in\Xi_k$.
	\end{itemize}
\end{lemma}
\begin{proof}
	Algorithm \ref{alg:nn2pwl} always terminates since all of its loops, which are originated from Algorithm \ref{alg:nn2pwl-inner} calls, range over some finite set and all recursive calls in Algorithm \ref{alg:nn2pwl-inner} increments the index of the input layer $L_i$, which will eventually reach the last layer $L_\Lambda$ and break the recursion by falsifying the conditional statement in line \ref{alg:nn2pwl-inner-firstIf} of Algorithm \ref{alg:nn2pwl-inner}. Let $\Xi_k$ be an entry in $\Xi_N$; all inequalities added to regions in $\Xi_k$ determine half-spaces in $[0,1]^{|L_0|}$. Indeed, this is the case in the first step of the construction of regions (Algorithm \ref{alg:nn2pwl}, line \ref{alg:nn2pwl-generic-region}). This is also the case for the remaining half-spaces, whose corresponding inequalities are recursively added in lines \ref{alg:nn2pwl-inner-halfspaces1}, \ref{alg:nn2pwl-inner-omega1}, \ref{alg:nn2pwl-inner-omega2} and \ref{alg:nn2pwl-inner-omega3} of Algorithm \ref{alg:nn2pwl-inner} and depend on its input tuple of linear functions $f$, which, in turn, are inductively defined over $[0,1]^{|L_0|}$: first by $\pi$ (Algorithm \ref{alg:nn2pwl}, line \ref{alg:nn2pwl-projection-tuple}) in the first call of \textsc{nn2pwl-r} (Algorithm \ref{alg:nn2pwl}, line \ref{alg:nn2pwl-recursive-call}); then, by $f^i_{\bowtie}$, for $i\in\{2,\ldots,\Lambda\}$ (Algorithm \ref{alg:nn2pwl-inner}, line \ref{alg:nn2pwl-inner-functions1}) in the following $\Lambda-1$ calls of \textsc{nn2pwl-r} (Algorithm \ref{alg:nn2pwl-inner}, line \ref{alg:nn2pwl-inner-recursive-call1}). Let $\mathbf{x}\in [0,1]^{|L_0|}$; note that among the possibilities for inequalities to be added in each step of the construction of regions, there is certainly one that is satisfied by $\mathbf{x}$. Thus, with the suitable configuration of symbols, there is a region $\Omega$ of $\Xi_k$ built such that $\mathbf{x}\in\Omega$. Now, in the construction of two regions $\Omega'$ and $\Omega''$ of $\Xi_k$, with $\Omega'\neq\Omega''$, there is some step where the added inequalities differ for $\Omega'$ and $\Omega''$ for the first time. Such differing inequalities guarantee that $\inte{\Omega'}\cap\inte{\Omega''}=\emptyset$, whether they appear in an intermediate step or the final one.
\end{proof}

\begin{lemma}\label{thm:nn2pwl-eval}
	Let a ReLU--TId neural network $N$, for which $\mathcal{L}_N = \{ L_0,\ldots,L_\Lambda \}$, be given as input to Algorithm~\ref{alg:nn2pwl} and let $\Xi_k$ be an entry in the outputted tuple $\Xi_N$ for which $\tuple{p,\Omega}\in\Xi_k$. If $\mathbf{x}\in\Omega$, then $N(\mathbf{x})_k=p(\mathbf{x})$.
\end{lemma}
\begin{proof}
	Let $\mathbf{x}\in\Omega$. Note that $\Omega = \Omega_{[0,1]} \cap \Omega^1_{\bowtie} \cap \cdots \cap \Omega^{\Lambda-1}_{\bowtie} \cap \Omega_{\bowtie}$ and $\Omega_{[0,1]} \supseteq \Omega^1_{\bowtie} \supseteq \cdots \supseteq \Omega^{\Lambda-1}_{\bowtie} \supseteq \Omega_{\bowtie}$, where $\Omega^i_{\bowtie}$ is such that $\bowtie\in\{\leq,\geq\}^{|L_i|}$ and $\Omega_{\bowtie}\in\{\Omega_\leq,\Omega_\lessgtr,\Omega_\geq\}$. Then, as $\mathbf{x}\in\Omega_{[0,1]}$, $\mathbf{x}\in [0,1]^{|L_0|}$. Also, as $\mathbf{x}\in\Omega^i_{\bowtie}$, for $i\in\{1,\ldots,\Lambda-1\}$, the tuples of functions $f^i_{\bowtie}$ defined in line \ref{alg:nn2pwl-inner-functions1} of Algorithm~\ref{alg:nn2pwl-inner}, given as arguments in the recursive call of \textsc{nn2pwl-r}, are such that $f^i_{\bowtie}(\mathbf{x}) = L_i\circ\cdots\circ L_1(\mathbf{x})$. Indeed, as $\mathbf{x}\in\Omega^1_{\bowtie}$, $\mathbf{x}$ satisfies the inequalities
	\[ f_1^1(\mathbf{x})=f_1^1\circ\pi(\mathbf{x}) \bowtie_1 0, ~~~\ldots,~~~ f_{|L_1|}^1(\mathbf{x})=f_{|L_1|}^1\circ\pi(\mathbf{x}) \bowtie_{|L_1|} 0.  \]
	Then, we have that
	\[ f^1_{\bowtie}(\mathbf{x}) = \tuple{\chi(\bowtie_1) \cdot f_1^1\circ\pi(\mathbf{x}), \ldots, \chi(\bowtie_{|L_1|}) \cdot f_{|L_1|}^1\circ\pi(\mathbf{x})} = \tuple{\mathrm{ReLU}(f_1^1(\mathbf{x})), \ldots, \mathrm{ReLU}(f_{|L_1|}^1(\mathbf{x}))} = L_1(\mathbf{x}). \]
	Now, let us assume that $f^i_{\bowtie}(\mathbf{x})=L_i\circ\cdots\circ L_1(\mathbf{x})$, for $\mathbf{x}\in\Omega$. As, in particular, $\mathbf{x}\in\Omega^{i+1}_{\bowtie}$, $\mathbf{x}$ satisfies the inequalities
	\[ f^{i+1}_1\circ f^i_{\bowtie}(\mathbf{x})=f^{i+1}_1\circ L_i\circ\cdots\circ L_1(\mathbf{x}) \bowtie_1 0, ~~~\ldots,~~~ f^{i+1}_{|L_{i+1}|}\circ f^i_{\bowtie}(\mathbf{x})=f^{i+1}_{|L_{i+1}|}\circ L_i\circ\cdots\circ L_1(\mathbf{x}) \bowtie_{|L_{i+1}|} 0,  \]
	it follows that
	\begin{align*}
		f^{i+1}_{\bowtie}(\mathbf{x}) &= \tuple{\chi(\bowtie_1) \cdot f_1^{i+1}\circ L_i\circ\cdots\circ L_1(\mathbf{x}), ~~~\ldots,~~~ \chi(\bowtie_{|L_1|}) \cdot f_{|L_1|}^{i+1}\circ L_i\circ\cdots\circ L_1(\mathbf{x})} \\
		&= \tuple{\mathrm{ReLU}(f_1^{i+1}\circ L_i\circ\cdots\circ L_1(\mathbf{x})), ~~~\ldots,~~~ \mathrm{ReLU}(f_{|L_1|}^{i+1}\circ L_i\circ\cdots\circ L_1(\mathbf{x}))} \\
		&= L_{i+1}\circ\cdots\circ L_1(\mathbf{x}).
	\end{align*}
	Finally, in case $\Omega_{\bowtie}=\Omega_\lessgtr$ (Algorithm \ref{alg:nn2pwl-inner}, line \ref{alg:nn2pwl-inner-omega2}), as, in particular, $\mathbf{x}\in\Omega_\lessgtr$, we have that
	\[ 0 \leq f_k^\Lambda\circ L_{\Lambda-1}\circ\cdots\circ L_1(\mathbf{x}) \leq 1. \]
	Therefore,
	\[ N(\mathbf{x})_k = \mathrm{TId}( f_k^\Lambda\circ L_{\Lambda-1}\circ\cdots\circ L_1(\mathbf{x}) ) = f_k^\Lambda\circ L_{\Lambda-1}\circ\cdots\circ L_1(\mathbf{x}) = p(\mathbf{x}). \]
	The other cases where $\Omega$ is either $\Omega_{\leq}$ or $\Omega_{\geq}$ are similar.
\end{proof}

\begin{theorem}[Correctness] \label{thm:nn2pwl-correct}
	Let a ReLU--TId neural network $N$, for which $\mathcal{L}_N = \{ L_0,\ldots,L_\Lambda \}$, be given as input to Algorithm~\ref{alg:nn2pwl} and let $\Xi_N=\tuple{\Xi_1,\ldots,\Xi_n}$ be its output. Then, each entry $\Xi_k$ in $\Xi_N$ codifies a rational McNaughton function in the pre-closed regional format which is exactly the function computed by $N$ through the path to its $k$-th output node.
\end{theorem}
\begin{proof}
	By construction and Lemma \ref{thm:nn2pwl-omega}, regions in $\Xi_k$ comply to the properties of pre-closed regional format. Lemma \ref{thm:nn2pwl-eval} establishes that evaluation via $\Xi_k$ is the same as via the $k$-th output node. Since the function computed via the $k$-th output node is a composition of continuous functions---both linear functions associated to nodes of $N$ and activation functions---, it is a continuous function. Therefore, entries in the tuple $\Xi_N$ codify continuous functions which are rational McNaughton functions.
\end{proof}

\begin{corollary}
	ReLU--TId neural networks are $\nu$-rational McNaughton neural networks, where $\nu=|L_\Lambda|$.
\end{corollary}

\subsection{Decreasing the execution time of the base algorithm}
The base algorithm just introduced has the downside to be exponential in the number of nodes of a given neural network. For a neural network $N$ with $\mathcal{L}_N = \{ L_0,\ldots,L_\Lambda \}$, we have seen that it computes $3|L_\Lambda| \times 2^{|L_1|+\cdots+|L_{\Lambda-1}|}$ regions. However, many of such regions may be the empty set, which makes the outputs of the base algorithm examples of degenerate codification in pre-closed regional format.

\begin{example}
	In a configuration of symbols where we associate $\geq$ to $n_1^1$ and $\leq$ to $n_2^1$ in the neural network $E$ in Example \ref{ex:example}, the corresponding inequalities $\frac{4}{3}x_1-x-2\geq 0$ and $x_1-x_2+\frac{1}{2}\leq 0$ together, related to the first layer $L_1$, determine the empty set.
\end{example}

In the step-by-step construction of a region $\Omega=\emptyset$ by the base algorithm, there is some step from the second when equations are added to the current polyhedron turning it into the empty set. In view of that, the first addition proposed for decreasing the execution time of the base algorithm consists in:
\begin{itemize}
	\item Conditioning the call of \textsc{nn2pwl-r} in line \ref{alg:nn2pwl-inner-recursive-call1} of Algorithm \ref{alg:nn2pwl-inner} by placing it within the scope of an if-statement that verifies whether $\Omega^i_{\bowtie} \neq \emptyset$;
	\item Conditioning the addition of new pairs $\tuple{p,\Omega_{\bowtie}}$, for all ${\bowtie}\in\{\leq,\lessgtr,\geq\}$, to tuple $\Xi_N$ in lines \ref{alg:nn2pwl-inner-xi1}, \ref{alg:nn2pwl-inner-xi2} and \ref{alg:nn2pwl-inner-xi3} of Algorithm \ref{alg:nn2pwl-inner} by placing these commands within the scope of if-statements that verify whether $\Omega_\leq\neq\emptyset$, $\Omega_\lessgtr\neq\emptyset$ and $\Omega_\geq\neq\emptyset$.
\end{itemize}

Verifying whether $\Omega_{\bowtie} \neq \emptyset$ might significantly decrease the running time of the translation algorithm in practice. Indeed, each true statement $\Omega_{\bowtie} \neq \emptyset$ occurring in the the $i$-th step of the construction of regions, for $i\in\{1,\ldots,\Lambda-1\}$, avoids a call of \textsc{nn2pwl-r} that, in the pure base algorithm, would yield the computation of $3|L_\Lambda| \times 2^{|L_{i+1}|+\cdots+|L_{\Lambda-1}|}$ pairs $\tuple{p,\Omega}$. On the other hand, verifying whether $\Omega_\leq\neq\emptyset$, $\Omega_\lessgtr\neq\emptyset$ or $\Omega_\geq\neq\emptyset$ in the last step of the construction of regions only prevents the algorithm to add pairs with empty regions to the regional format codification, which, nevertheless, makes the final representative tuple $\Xi_N$ smaller.

A possible way to verify whether a polyhedron $\Omega$ given as in \eqref{eq:region} is nonempty is by applying the known polynomial techniques used to verify whether a linear optimization program constrained by $\Omega$ is feasible \cite{BT1997}.

For another method for easing the execution time of the base algorithm, observe that, for a layer $L_i$, for $i\in\{1,\ldots,\Lambda-1\}$, each of the hyperplanes $f_j^i(\mathbf{x})=0$ related to nodes $n_j^i$ of $L_i$, for $j\in\{1,\ldots,|L_i|\}$, divides the euclidean space $\mathbb{R}^{|L_0|}$ in two half-spaces determined by the inequalities $f_j^i(\mathbf{x})\geq 0$ and $f_j^i(\mathbf{x})\leq 0$. Each of these inequalities is added to half of the $2^{|L_i|}$ polyhedra generated in the first for-loop of Algorithm \ref{alg:nn2pwl-inner} (lines \ref{alg:nn2pwl-inner-For} to \ref{alg:nn2pwl-inner-For-end}); these are the polyhedra generated in the $i$-th step of the construction of regions. Now, note that if the hyperplane $f_j^i(\mathbf{x})=0$ does not intercept the interior of the unit cube $[0,1]^{|L_0|}$, half of the new generated polyhedra are certainly empty. For instance, the hyperplane $x_1 + x_2 - 2 = 0$ does not intercept $[0,1]^{|L_0|}$. Thus, although the half-space given by $x_1 + x_2 - 2 \leq 0$ contains the entire unit cube $[0,1]^{|L_0|}$, the half-space given by $x_1 + x_2 - 2 \geq 0$ does not intersect it; so, if $x_1 + x_2 - 2 \geq 0$ is added to a polyhedron in some step of the construction of regions by the base algorithm, the regions generated from such polyhedron will be the empty set.

Thus, for the step related to layer $L_i$, for $i\in\{1,\ldots,\Lambda\}$, in the construction of regions, the proposed method consists in building a set $\mathbf{I}\subseteq\{\leq,\geq\}^{|L_i|}$ to be iterated instead of the set $\{\leq,\geq\}^{|L_i|}$ in the for-loop beginning in line \ref{alg:nn2pwl-inner-For} of Algorithm \ref{alg:nn2pwl-inner}, so avoiding the generation of empty polyhedra. For that, we compute $\mathbf{I}=\mathbf{I}_1\times\cdots\times\mathbf{I}_{|L_i|}$ where, for $j\in\{1,\ldots,|L_i|\}$,
\begin{align*}
	\mathbf{I}_j = \left\{
	\begin{array}{ll}
		\{\geq,\leq\}, & \text{if $f_j^i(\mathbf{x})=0$ intercepts the interior of $[0,1]^{|L_0|}$} \\
		\{\leq\}, & \text{if $f_j^i(\mathbf{x})\leq 0$ contains the entire $[0,1]^{|L_0|}$} \\
		\{\geq\}, & \text{if $f_j^i(\mathbf{x})\geq 0$ contains the entire $[0,1]^{|L_0|}$}
	\end{array} \right.
\end{align*}
Determining which is the case for each $\mathbf{I}_j$ may be done by solving both of the following maximization and minimization linear programs, which are known to be solvable in polynomial time \cite{BT1997}:
\begin{align*}
	\begin{array}{lll}
		\max / \min	& f_j^i(\mathbf{x}) \\
		\mbox{subject to} & [0,1]^{|L_0|}
	\end{array}
\end{align*}
Let $M$ and $m$ respectively be the maximum and the minimum optimum values of the linear programs above. Then: if $M\geq 0$ and $m\leq 0$ or if $M\leq 0$ and $m\geq 0$, $f_j^i(\mathbf{x})=0$ intercepts $[0,1]^{|L_0|}$; if $M\geq 0$ and $m\geq 0$, $f_j^i(\mathbf{x})\geq 0$ contains $[0,1]^{|L_0|}$; and if $M\leq 0$ and $m\leq 0$, $f_j^i(\mathbf{x})\leq 0$ contains $[0,1]^{|L_0|}$. The overall execution time of the translation algorithm, even with an additional routine for building $\mathbf{I}$, might be significantly smaller than the time for the original base algorithm. In fact, let $J\subseteq\{1,\ldots,|L_i|\}$ be the set of indexes such that $\mathbf{I}_j\neq\{\leq,\geq\}$ if, and only if, $j\in J$; then, the for-loop beginning in line \ref{alg:nn2pwl-inner-For} of Algorithm \ref{alg:nn2pwl-inner} has $2^{|L_i|-|J|}$ iterations instead of $2^{|L_i|}$.

Combining both of the methods described in this section with the base translation algorithm makes it compute exactly the same pairs $\tuple{p,\Omega}$ that it would compute without such methods with the exception of the ones for which $\Omega=\emptyset$. Therefore, we are able to establish the following result.

\begin{theorem}
	Replacing the routine \textsc{nn2pwl-r} for a version of it that includes the methods proposed in this section maintains the correctness of Algorithm \ref{alg:nn2pwl} established in Theorem \ref{thm:nn2pwl-correct}.
\end{theorem}

\section{Experiments and Results}
\label{sec:experiments}
We perform experiments for measuring the \emph{complexity} of pre-closed regional format encodings of randomly generated ReLU--TId neural networks by counting the number of nonempty regions in them. All weights of the neural networks have the form $i+d$, where both $i$ and $d$ are uniformly generated from $\{-1,0,1\}$ and $[0,1)$, respectively.

For each encoding in pre-closed regional format, we also evaluate its \emph{degree of satisfiability} of the lattice property by counting the number of pairs of regions $\tuple{\Omega_i,\Omega_j}$ for which there is no linear piece $p_k$ such that $p_i$ is above $p_k$ over $\Omega_i$ and $p_k$ is above $p_j$ over $\Omega_j$, that is the number of pairs $\tuple{\Omega_i,\Omega_j}$ that falsifies lattice property. If the counting is $0$, such an encoding completely satisfies lattice property; the higher the count the further from satisfying lattice property the encoding is.

Implementations of \textsc{nn2pwl}, including the methods for decreasing its execution time, and a neural network generator were developed for the experiments; the source code is publicly available.\footnote{\url{http://github.com/spreto/reluka}}

In the first batch of experiments, for a fixed value $h$, ReLU--TId neural networks with $h$ input neurons, $h$ neurons in each hidden layer and one output neuron are generated. Such random generation is done in such a way that the neural networks are partitioned in $L$ classes, each containing $n$ neural networks with $l$ hidden layers, for $l\in\{1,\ldots,L\}$. We ran such experiment for two parameter setups: $h=4$, $L=6$, $n=50$ and $h=5$, $L=10$, $n=25$. Figure \ref{fig:experiments1} depicts the average number of regions extracted from the neural networks in each class of $l$ hidden layers.

\begin{figure}
	\centering
	\includegraphics[width=250pt]{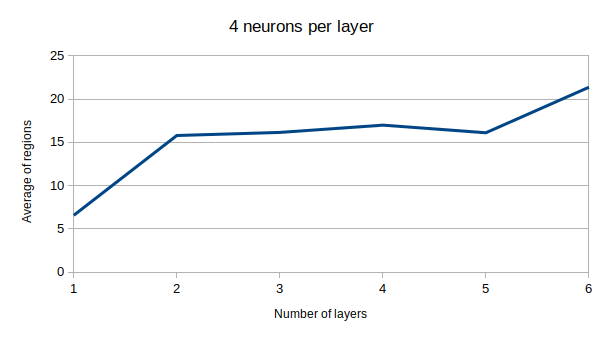}

	\includegraphics[width=250pt]{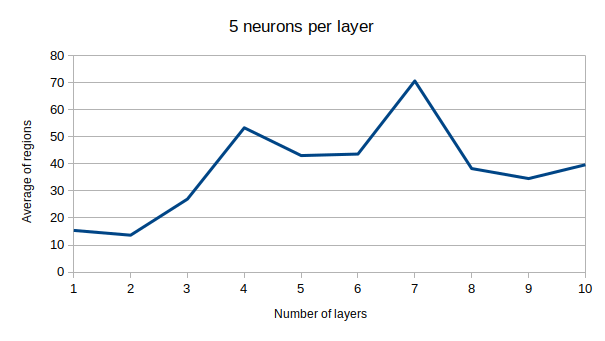}
	\caption{Experiments increasing the number of layers}
\label{fig:experiments1}
\end{figure}

In order to analyze whether the results of the previous experiments depend on the distribution of neurons per layer, in the second batch of experiments, ReLU--TId neural networks with a fixed number $l$ of hidden layers and one output neuron are generated. Now, the randomly generated neural networks are partitioned in $M$ classes, each containing $n$ neural networks with $m$ input neurons and $m$ neurons in each of their hidden layers, for $m\in\{1,\ldots,M\}$. We ran such experiment for two parameter setups: $l=4$, $M=6$, $n=50$ and $l=5$, $M=10$, $n=25$. Figure \ref{fig:experiments2} depicts the average number of regions extracted from the neural networks in each class of $m$ neurons in the input layer and per hidden layer.

\begin{figure}
	\centering
	\includegraphics[width=250pt]{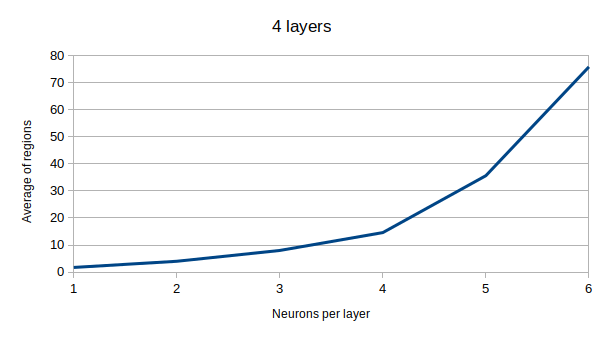}
	
	\includegraphics[width=250pt]{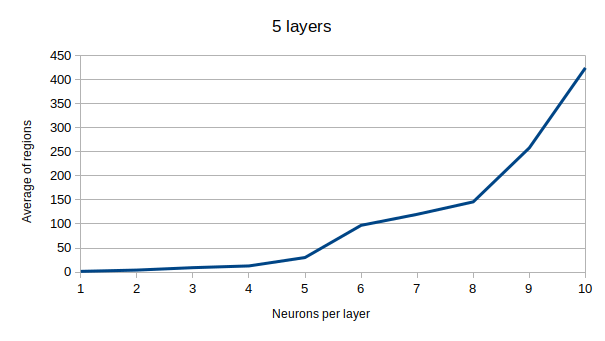}
	\caption{Experiments increasing the number of neurons per layer}
	\label{fig:experiments2}
\end{figure}

Note that the first experiment in both batches of experiments are related: for $l=m$, neural networks in a class with $l$ layers (first experiment, first batch) have the same number of neurons than the neural networks in a class with $m$ nodes per layer (first experiment, second batch). The same relation may be seen between the second experiments of each batch.

In all experiments, we may see that the average number of regions increases as long as the number of neurons increases. However, while such variation in the number of regions is smooth for varying the number of layers, a sharp variation may be perceived for varying the number of neurons per layer. A neural network with $5$ hidden layers and $10$ neurons in each of them ($50$ neurons in all hidden layers) achieved the maximum number of $1852$ regions among all neural networks generated. For comparison, among the neural networks with $50$ neurons distributed in $10$ hidden layers ($5$ neurons per hidden layer), the maximum number of regions achieved is $228$. And among all the neural networks with more than $5$ hidden layers, but only $5$ neurons in each of them, the maximum number of regions achieved is $446$ (in a neural network with $7$ hidden layers).

Regarding lattice property, among all $1100$ ReLU--TId neural networks that were generated in all experiments, only one failed to satisfy it. Such neural network has $5$ neurons in each of its $5$ hidden layers ($25$ neurons in all hidden layers) and its pre-closed regional format encoding has $91$ regions and fails to fulfill lattice property for $36$ pairs of regions $\tuple{\Omega_i,\Omega_j}$.

\section{Conclusions}
\label{sec:conclusions}
We have proposed an algorithm for translating ReLU--TId neural networks into the pre-closed regional format, which is a more interpretable representation than the traditional graph one. We also proposed methods for decreasing the computation time of the base algorithm and proved that ReLU--TId neural networks are $\nu$-rational McNaughton neural networks.

Empirically, we measured the complexity of pre-closed regional format encodings of randomly generated ReLU--TId neural networks by counting the number of nonempty regions in such encodings. We could verify a bigger increase in the number of regions in the encodings with wider, but fewer, layers than in the encodings with more, but thinner, layers. The fast increase of curves in Figure \ref{fig:experiments2}, related to the variation in the size of a fixed number of layers, points to the high complexity of regional representation. Therefore, the reported results foresee scaling issues in the regional representation of real-world neural networks, which often are larger than those generated in our investigation.

We have also investigated the degree of satisfiability of the lattice property by the neural networks generated in our experiments. The results empirically indicate that the outputs of \textsc{nn2pwl} lacking lattice property are a very rare event. Only one of the neural networks generated do not fulfill such a property.

For the future, approximate and less complex regional representations might be pursued. A possible path is to establish the reasonability of allowing encodings not satisfying lattice property as approximations of neural networks. From an exact perspective, one might investigate efficient procedures for turning a rational McNaughton function encoding in pre-closed regional format into closed regional format.

\section*{Funding}
This work was carried out at the Center for Artificial Intelligence (C4AI-USP), with support by the S\~ao Paulo Research Foundation (FAPESP) [grant \#2019/07665-4] and by the IBM Corporation. This study was financed in part by the S\~ao Paulo Research Foundation (FAPESP) [grants \#2021/03117-2 to S.P., \#2015/21880-4 and \#2014/12236-1 to M.F.]; and the National Council for Scientific and Technological Development (CNPq) [grant PQ 303609/2018-4 to M.F.].

\bibliographystyle{eptcs}
\bibliography{references}
\end{document}